\documentclass[journal]{IEEEtran}

\ifCLASSINFOpdf
\else
   \usepackage[dvips]{graphicx}
\fi
\usepackage{url}
\usepackage{graphicx}
\usepackage{algorithm}
\usepackage{algpseudocode}
\usepackage{amssymb}
\usepackage{xcolor}
\usepackage{tabularx}
\usepackage{booktabs}
\usepackage{mathtools}
\usepackage{bm}

\newcommand{\E}{\mathbb{E}}
\newcommand{\tr}{\operatorname{tr}}

\newcommand{\Var}{\operatorname{Var}}
\newcommand{\CRB}{\mathrm{CRB}}

\newtheorem{proposition}{Proposition}
\newtheorem{theorem}{Theorem}
\newtheorem{remark}{Remark}

\usepackage{amsmath,amsfonts}
\usepackage{array}
\usepackage[caption=false,font=normalsize,labelfont=sf,textfont=sf]{subfig}
\usepackage{textcomp}
\usepackage{stfloats}
\usepackage{url}
\usepackage{verbatim}
\usepackage{hyperref}
\hypersetup{
    colorlinks=true,
    linkcolor=blue,
    citecolor=blue,
    urlcolor=blue}

\begin{document}

\title{Continuous Intra-Symbol Phase Noise Tracking for THz OFDM via Polynomial Reconstruction}

\author{Sawatsakorn Chaiyasoonthorn, Ura Klongklaew, and Phichai Youplao
\thanks{This is the accepted version of the article published in IEEE Signal Processing Letters. The final version is available at https://doi.org/10.1109/LSP.2026.3711541. © 2026 IEEE.}
}

\maketitle

\begin{abstract}
Terahertz (THz) communication systems for sixth-generation (6G) networks are severely impaired by Wiener phase noise (WPN), whose innovation variance at sub-THz carriers is substantially larger than in millimeter-wave 5G systems. Conventional common-phase-error (CPE) compensation applies a single phase rotation per OFDM symbol and becomes inadequate when the phase trajectory varies significantly within the symbol duration. This letter proposes continuous phase trajectory reconstruction (CPTR), a closed-form intra-symbol phase noise tracking method that reconstructs the sample-level phase trajectory from pilot observations via least-squares polynomial fitting with $\mathcal{O}(N_p+N)$ complexity. We characterize the polynomial approximation error under WPN and derive the Cram\'{e}r--Rao bound (CRB) for polynomial phase coefficient estimation, showing that CPTR is minimum-variance unbiased within the polynomial surrogate model. Simulations at 300~GHz with \textit{N}~=~1024 and 16-QAM show that CPTR remains within 0.2~dB of the CRB across SNR~=~10--45~dB while achieving significantly lower complexity than Kalman-based tracking and substantial BER gains over CPE, linear interpolation, and cubic spline methods.
\end{abstract}


\IEEEpeerreviewmaketitle

\section{Introduction}
\label{sec:intro}

\IEEEPARstart{T}{erahertz} communications (0.1--10~THz) are a key enabler of sixth-generation (6G) networks \cite{Han2022, Sharma2025}, supporting aggregate data rates beyond 1~Tbps through ultra-wide bandwidths \cite{Akyildiz2020, Rappaport2019, Sarieddeen2021}. However, local-oscillator (LO) instability at these frequencies introduces WPN with innovation variance $\sigma_{\Delta}^{2}\approx 4\pi^{2}f_{c}^{2}c_0T_s$, where $c_0$ is the Lorentzian linewidth coefficient and $T_s$ is the sampling period \cite{Khanzadi2014, Bicais2019}. At $300$~GHz and $20$~GHz bandwidth, the resulting intra-symbol phase drift makes inter-carrier interference (ICI) a dominant impairment \cite{Pollet1995, Wu2004}.

Conventional CPE compensation applies a single phase rotation per OFDM symbol, effectively treating the phase trajectory as constant within each OFDM symbol~\cite{Petrovic2007, Stefanatos2017}. Polynomial phase modeling has been considered heuristically~\cite{Corvaja1997, Giannakis1998, Do2025} without rigorous characterization of approximation error or statistical efficiency under WPN. Kalman-based methods~\cite{Rabiei2010, Mehrpouyan2012} improve estimation accuracy but incur $\mathcal{O}(Nd^{2})$ complexity, where $d$ is the state dimension. Pilot-aided phase-noise estimation methods~\cite{Mathecken2017, Chen2025} remain iterative or frequency-domain oriented, while deterministic interpolation schemes (e.g., linear or cubic spline) suffer from boundary artifacts and lack statistical optimality. Alternatively, Bayesian approaches, including Wiener interpolation filters~\cite{Chang2025,Bello2023}, exploit prior statistical information of the phase-noise process for estimation. Existing methods, however, do not jointly provide closed-form low-complexity estimation, rigorous statistical characterization, and communication-level BER validation. To bridge this gap, this letter proposes a deterministic polynomial least-squares framework for continuous intra-symbol phase reconstruction. The main contributions are:
\begin{enumerate}
  \item \textit{Continuous intra-symbol phase reconstruction:}
        CPTR reconstructs the sample-level phase trajectory
        via a single least-squares solve of dimension
        $(P{+}1)\!\times\!(P{+}1)$, independent of the FFT size $N$.

  \item \textit{Approximation error characterization:}
        We derive an asymptotic upper bound on the polynomial
        approximation error under WPN
        (Proposition~\ref{prop:approx}).

  \item \textit{CRB analysis and MVU optimality:}
        We derive the CRB for polynomial phase coefficient
        estimation and prove that CPTR is MVU within the
        polynomial surrogate model
        (Theorems~\ref{thm:crb}--\ref{thm:mvu}).

  \item \textit{Communication-level validation:}
        Extensive simulations at 300~GHz with 16-QAM demonstrate
        BER gains over CPE, interpolation,
        cubic-spline, and EKF baselines.
\end{enumerate}

\section{System Model}
\label{sec:model}

Consider a THz OFDM system~\cite{Tarboush2021}
with $N$ subcarriers, CP length $N_{\mathrm{cp}}$,
and symbol duration
$T_{\mathrm{sym}}=(N+N_{\mathrm{cp}})T_s$.
The transmitted baseband signal is
\begin{equation}\label{eq:tx}
  x(t)=\frac{1}{\sqrt{N}}
  \sum_{k=0}^{N-1}X_k e^{j2\pi k\Delta f t},
  \quad 0\le t<NT_s,
\end{equation}
where $X_k\in\mathcal{X}$ and
$\Delta f=1/(NT_s)$.
After LO mixing, the received signal is
\begin{equation}\label{eq:rx}
  y(t)=x(t)e^{j\phi(t)}+w(t),
\end{equation}
where $w(t)\sim\mathcal{CN}(0,N_0)$ and $\phi(t)$ denotes the LO phase noise. Although wideband oscillators may exhibit both correlated Wiener and uncorrelated Gaussian components~\cite{Khanzadi2014,Chang2025}, we focus on the Wiener component to enable closed-form polynomial reconstruction and CRB characterization. The discrete-time phase noise process $\phi[n]\triangleq\phi(nT_s)$ follows
\begin{equation}\label{eq:wpn}
  \phi[n]=\phi[n-1]+\delta[n],\quad
  \delta[n]\sim\mathcal{N}(0,\sigma_{\Delta}^{2}),
\end{equation}
with $\sigma_{\Delta}^{2}=4\pi^{2}f_c^{2}c_0T_s$. The intra-symbol phase excursion variance is
\begin{equation}\label{eq:excursion}
  \Var[\phi[N-1]-\phi[0]]
  =(N-1)\sigma_{\Delta}^{2}.
\end{equation}
At $300$~GHz with $N=1024$ and $T_s=50$~ps,
the excursion standard deviation reaches
approximately $0.43$~rad, exceeding the validity
range of CPE-only compensation.

After CP removal and FFT, the received subcarrier is
\begin{equation}\label{eq:yk}
  Y_k=H_kX_kI_0
      +\sum_{\substack{m=0\\m\ne k}}^{N-1}
      H_mX_mI_{k-m}+W_k,
\end{equation}
where $H_k \in \mathbb{C}$ is the channel coefficient at subcarrier~$k$,
assumed known at pilot positions, following~\cite{Rabiei2010, Chen2025}, and
\begin{equation}
  I_\ell=\frac{1}{N}\sum_{n=0}^{N-1}
  e^{j\phi[n]}e^{-j2\pi \ell n/N}
\end{equation}
denotes the phase-noise coefficient.
Assuming $N_p$ equally spaced pilot subcarriers
$\mathcal{P}=\{p_0,\ldots,p_{N_p-1}\}$,
the pilot phase observation is
\begin{equation}\label{eq:obs}
  \tilde{\phi}_i=
  \angle\!\left(
  \frac{Y_{p_i}}{H_{p_i}X_{p_i}}
  \right),
  \quad i=0,\ldots,N_p-1.
\end{equation}
Under moderate-to-high SNR, linearization yields
an approximately Gaussian observation model with
variance $\sigma_n^{2}=N_0/(2E_s)$.

\section{Proposed CPTR Estimator}
\label{sec:cptr}

\subsection{Polynomial Phase Trajectory Model}

Over a single OFDM symbol, the WPN trajectory evolves smoothly
and is dominated by low-frequency components, motivating a
low-order polynomial surrogate.

\begin{proposition}[Polynomial Approximation Error Bound]
\label{prop:approx}
Let $\hat{\phi}(n)=\sum_{p=0}^{P}a_p\psi_p(n)$, $\psi_p(n)=(n/N)^p$, denote the degree-$P$ polynomial approximation of the WPN trajectory $\phi[n]$. Then, asymptotically for $P\ll N$, the mean-squared approximation error satisfies
\begin{equation}\label{eq:bound}
  \frac{1}{N}\sum_{n=0}^{N-1}
  \E\{|\varepsilon[n]|^{2}\}
  \lesssim
  C\,\sigma_{\Delta}^{2}N(P+1)^{-1},
\end{equation}
where $\varepsilon[n]=\phi[n]-\hat{\phi}(n)$ and $C>0$ is independent of $N$ and depends only on the polynomial basis conditioning. This result indicates that increasing the polynomial order reduces the approximation error algebraically, albeit with diminishing returns for large $P$.
\end{proposition}

A cubic model ($P=3$) captures more than $99\%$ of the
intra-symbol phase variance at $300$~GHz while maintaining
numerical stability.

\subsection{Closed-Form Least-Squares Estimator}

Stacking the pilot observations into
$\tilde{\bm{\phi}}\in\mathbb{R}^{N_p}$ yields
\begin{equation}\label{eq:linmod}
  \tilde{\bm{\phi}}
  =\bm{\Psi}\bm{a}+\bm{\varepsilon},
\end{equation}
where
$\bm{a}=[a_0,\ldots,a_P]^\top$,
$[\bm{\Psi}]_{i,p}=\psi_p(p_i/N)$,
and
$\bm{\varepsilon}\sim
\mathcal{CN}(\bm{0},\sigma_n^2\mathbf{I})$.
The CPTR estimate is
\begin{equation}\label{eq:cptr}
  \hat{\bm{a}}
  =
  (\bm{\Psi}^\top\bm{\Psi})^{-1}
  \bm{\Psi}^\top\tilde{\bm{\phi}}.
\end{equation}

The reconstructed phase trajectory is
$\hat{\phi}[n]=\bm{\psi}(n)^\top\hat{\bm{a}}$,
where
$\bm{\psi}(n)=[1,(n/N),\ldots,(n/N)^P]^\top$,
and compensation is applied as
$\tilde{y}[n]=y[n]e^{-j\hat{\phi}[n]}$.

\begin{remark}[Coefficient-Space Comparison]
\label{rem:coef}
To compare CPTR and EKF on a common basis
in the polynomial coefficient domain of the surrogate
model~\eqref{eq:linmod}, we define the following
projection metric used in Section~\ref{sec:sims}.
CPTR directly produces $\hat{\bm{a}}$ via~\eqref{eq:cptr}.
For EKF, which operates recursively in the polynomial
state space, the reconstructed trajectory is projected
onto the same basis as
\begin{equation}\label{eq:proj_ekf}
  \hat{\bm{a}}_{\mathrm{EKF}} =
  (\bm{\Psi}^\top\bm{\Psi})^{-1}\bm{\Psi}^\top
  \hat{\bm{\phi}}_{\mathrm{EKF}}\big|_{\mathcal{P}},
\end{equation}
where $\hat{\bm{\phi}}_{\mathrm{EKF}}|_{\mathcal{P}}$
denotes the EKF phase estimate evaluated at pilot
positions. The reference coefficients are obtained by projecting
the true WPN trajectory onto the polynomial basis:
\begin{equation}\label{eq:aLS}
  \bm{a}_{\mathrm{LS}} =
  (\bm{\Psi}^\top\bm{\Psi})^{-1}\bm{\Psi}^\top
  \bm{\phi}_{\mathrm{true}}\big|_{\mathcal{P}}.
\end{equation}
The normalized coefficient MSE is then
\begin{equation}\label{eq:coef_mse}
  \mathrm{MSE}_{\mathrm{coef}} =
  \frac{1}{P+1}\,
  \mathbb{E}\bigl\|\hat{\bm{a}}_{\mathrm{m}}
  - \bm{a}_{\mathrm{LS}}\bigr\|^{2},
\end{equation}
where $\hat{\bm{a}}_{\mathrm{m}} \in
\{\hat{\bm{a}},\,\hat{\bm{a}}_{\mathrm{EKF}}\}$.
Non-parametric methods (linear interpolation, cubic spline)
are evaluated only in the trajectory domain and are
omitted from the coefficient-space comparison.
\end{remark}

The dominant computational cost of CPTR consists of three stages. Forming the Gram matrix requires $\mathcal{O}(N_pP^2)$ operations, matrix inversion requires $\mathcal{O}(P^3)$, and trajectory reconstruction requires $\mathcal{O}(NP)$. Since $P\le4$ is fixed, the overall complexity scales as $\mathcal{O}(N_p+N)$. Table~\ref{tab:complexity} compares CPTR with existing methods.

\begin{table}[!t]
\renewcommand{\arraystretch}{1.15}
\caption{Computational Complexity Comparison per OFDM Symbol}
\label{tab:complexity}
\centering
\begin{tabularx}{\columnwidth}{p{1.2in}p{1.2in}X}
\toprule
\textbf{Method} & \textbf{Multiplications} & \textbf{Order} \\
\midrule
CPE (zeroth-order)       & $2N$               & $\mathcal{O}(N)$     \\
Linear interpolation     & $2N$               & $\mathcal{O}(N)$     \\
Cubic spline             & $\sim 8N_p+2N$     & $\mathcal{O}(N_p+N)$ \\
EKF \cite{Rabiei2010}    & $\sim c_1 d^2 N$   & $\mathcal{O}(N d^2)$ \\
\textbf{CPTR (proposed)} & $\bm{2N + 6N_p}$   & $\mathcal{O}(N_p+N)$ \\
\bottomrule
\multicolumn{3}{l}{\footnotesize $d$: EKF state dimension; spline uses natural boundary conditions.}
\end{tabularx}
\end{table}

\section{Cram\'{e}r--Rao Bound Analysis}
\label{sec:crb}

Under the linear-Gaussian model~\eqref{eq:linmod},
the Fisher information matrix (FIM) for $\bm{a}$ is
\begin{equation}\label{eq:fim}
  \mathbf{J}(\bm{a})
  =
  \frac{1}{\sigma_n^2}
  \bm{\Psi}^\top\bm{\Psi}.
\end{equation}

This reveals that pilot placement directly conditions estimation efficiency through the spectrum of the Gram matrix $\bm{\Psi}^\top\bm{\Psi}$. Uniform pilot placement therefore improves numerical conditioning and reduces estimator variance by minimizing correlation among polynomial basis functions.

\begin{theorem}[CRB for Polynomial Phase Estimation]
\label{thm:crb}
For $N_p\ge P+1$ uniformly spaced pilots,
the CRB on the $p$-th polynomial coefficient is
\begin{equation}\label{eq:crb_coef}
  \CRB(a_p)
  =
  \sigma_n^2
  \bigl[
  (\bm{\Psi}^\top\bm{\Psi})^{-1}
  \bigr]_{pp}.
\end{equation}
The average coefficient CRB is
\begin{equation}\label{eq:crb_avg}
  \CRB_{\mathrm{coef}}
  =
  \frac{\sigma_n^2}{P+1}
  \tr\!\left[
  (\bm{\Psi}^\top\bm{\Psi})^{-1}
  \right].
\end{equation}
\end{theorem}

\begin{theorem}[MVU Efficiency of CPTR]
\label{thm:mvu}
The proposed CPTR estimator~\eqref{eq:cptr} is unbiased and achieves
the CRB under the polynomial surrogate model:
\begin{equation}\label{eq:mvu}
  \mathrm{Cov}(\hat{\bm{a}})
  =
  \sigma_n^2
  (\bm{\Psi}^\top\bm{\Psi})^{-1}
  =
  \mathbf{J}^{-1}(\bm{a}).
\end{equation}
Hence, CPTR is MVU within the surrogate model.
\end{theorem}

\textit{Proof:}
Unbiasedness follows from
$\E[\hat{\bm{a}}]=\bm{a}$.
The covariance identity follows from the
Slepian--Bangs formula for linear-Gaussian models,
for which least-squares estimation attains the
CRB identically~\cite{Kay1993, Trees2001}.

\begin{remark}[Model Mismatch]
\label{rem:scope}
The MVU result applies only within the polynomial
surrogate model. For the true Wiener trajectory,
the residual approximation floor is
\begin{equation}\label{eq:floor}
  \sigma_\varepsilon^2
  \approx
  \frac{\kappa(P)\sigma_\Delta^2 N}{P+2},
\end{equation}
arising from the mismatch between the Wiener process and its finite-order polynomial approximation (Proposition~\ref{prop:approx}). Consequently, the CRB is tight only for coefficient estimation.

\end{remark}

\section{Simulation Results and Discussion}
\label{sec:sims}

All numerical simulations were conducted using MATLAB. Table~\ref{tab:params} summarizes the simulation parameters for a $300$~GHz THz OFDM system. Phase noise follows~\eqref{eq:wpn} with $c_0=10^{-18}$~s. The adopted pilot density $N_p/N=1/16$ provides a practical balance between estimation accuracy and spectral efficiency. Six methods are compared: no compensation (NC), single-rotation CPE, linear interpolation, cubic spline, EKF ($d=4$)~\cite{Rabiei2010}, and the CPTR ($P=3$).

Fig.~\ref{fig:traj} compares the reconstructed phase trajectories for the instantaneous WPN $\phi[n]$ at SNR~$=30$~dB. CPTR accurately captures the dominant low-frequency phase dynamics, whereas CPE cannot track intra-symbol variation. Linear interpolation introduces piecewise distortion between pilot pairs, while cubic spline exhibits mild boundary ringing. EKF also tracks the trajectory accurately, but with higher recursive complexity.

Fig.~\ref{fig:results}(a) illustrates trajectory MSE versus SNR. NC and CPE exhibit pronounced error floors due to uncompensated intra-symbol phase evolution. Linear interpolation and cubic spline substantially reduce the trajectory distortion, with linear interpolation achieving very low trajectory MSE at high SNR through dense local fitting. EKF also provides accurate recursive trajectory tracking, whereas CPTR exhibits a residual floor consistent with the approximation limit of Proposition~\ref{prop:approx}. To assess robustness under mixed Wiener--Gaussian phase noise, as characterized in~\cite{Khanzadi2014, Bello2023}, an independent Gaussian component with variance $\sigma_G^2=10^{-4}$~rad$^2$ is superimposed on the Wiener phase noise as an illustrative case. CPTR exhibits only a marginal increase in trajectory MSE.

\begin{table}[!t]
\renewcommand{\arraystretch}{1.15}
\caption{System Simulation Parameters ($10^4$ Monte Carlo trials)}
\label{tab:params}
\centering
\begin{tabular}{ll}
\toprule
\textbf{Parameter} & \textbf{Value} \\
\midrule
Carrier frequency $f_c$ / Bandwidth $B$ & $300\,\mathrm{GHz}$ / $20\,\mathrm{GHz}$ \\
FFT size $N$ / CP length $N_\mathrm{cp}$ / Pilots $N_p$ & $1024$ / $64$ / $64$ \\
Modulation / Polynomial order $P$ & $16$-QAM / $3$ (cubic) \\
Phase noise coeff. $c_0$ / $\sigma^2_\Delta$ & $10^{-18}\,\mathrm{s}$ / $1.78\times10^{-4}$~rad$^2$ \\
Intra-symbol excursion std & $\approx 0.43$~rad \\
\bottomrule
\end{tabular}
\end{table}

\begin{figure}[!t]
  \centering
  \includegraphics[width=\columnwidth]{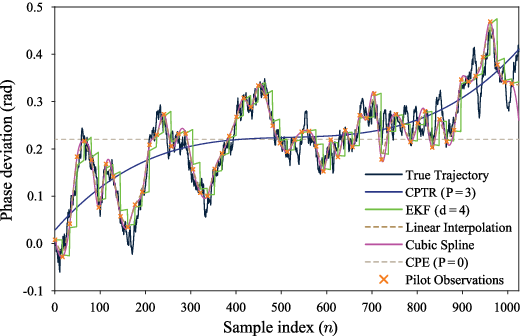}
  \caption{Intra-symbol phase-noise tracking over one OFDM symbol ($N=1024$) at SNR~$=30$~dB and $c_0=10^{-18}$~s. CPTR accurately reconstructs the phase trajectory. Linear interpolation introduces piecewise distortion, cubic spline exhibits mild boundary ringing, and CPE cannot track intra-symbol variation. EKF also achieves accurate tracking, but at higher complexity.}
  \label{fig:traj}
\end{figure}

Fig.~\ref{fig:results}(b) plots the coefficient MSE normalized to the CRB within the polynomial surrogate model, evaluating statistical estimation efficiency independently of the approximation error in Fig.~\ref{fig:results}(a). The CRB is computed for $P=3$, consistent with~\eqref{eq:crb_avg}. Per Remark~\ref{rem:coef}, non-parametric methods (linear interpolation, cubic spline) lack a global polynomial coefficient vector and are thus excluded from this comparison. CPTR remains tightly clustered around the CRB, with deviations below approximately $0.2$~dB across the simulated SNR range, indicating near-MVU efficiency within the polynomial surrogate model. EKF exhibits consistently larger deviations, reflecting sensitivity to process-noise covariance mismatch and recursive estimation errors, since it is designed for recursive trajectory tracking rather than optimality in the polynomial coefficient space defined by the surrogate model.

\begin{figure*}[!t]
  \centering
  \includegraphics[width=7.0in]{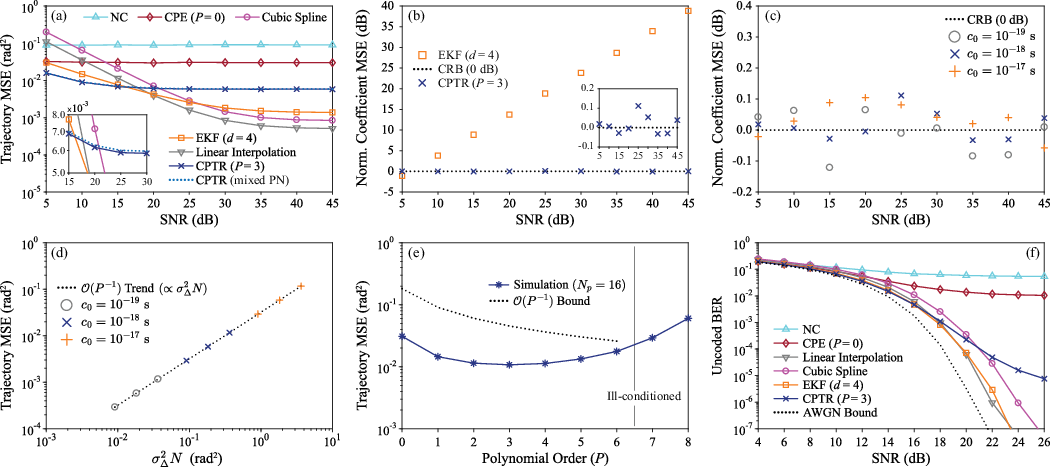}
  \caption{Simulation results at $f_c=300$~GHz, $N=1024$, $N_p=64$, 16-QAM, $c_0=10^{-18}$~s, unless otherwise stated. (a) Trajectory MSE vs.\ SNR; CPTR saturates at the approximation floor of Proposition~\ref{prop:approx}, while the dotted curve and inset show the marginal penalty under mixed Wiener--Gaussian phase noise (mixed PN) ($\sigma_G^2=10^{-4}$~rad$^2$). (b) Normalized coefficient MSE (dB above CRB) for EKF and CPTR (Remark~\ref{rem:coef}); CPTR stays within 0.2~dB of the CRB across all SNRs. (c) Normalized coefficient MSE of CPTR for different phase-noise severities $c_0$. (d) Trajectory MSE of CPTR vs.\ $\sigma_\Delta^2 N$, verifying the predicted $\mathcal{O}(\sigma_\Delta^2 N)$ scaling. (e) Trajectory MSE vs.\ polynomial order $P$ at $N_p=16$ and SNR~$=15$~dB. MSE decreases up to $P=3$, then increases for $P\ge6$ due to Vandermonde ill-conditioning; dashed line: asymptotic $\mathcal{O}(P^{-1})$ bound. (f) Uncoded BER vs.\ SNR; CPTR substantially outperforms CPE with a residual high-SNR floor from polynomial approximation error.}
    
  \label{fig:results}
\end{figure*}

Fig.~\ref{fig:results}(c) illustrates the evaluation of CPTR across $c_0\in\{10^{-19},10^{-18},10^{-17}\}$~s. The normalized coefficient MSE remains tightly clustered around the CRB, with deviations below approximately $\pm0.1$~dB and no systematic dependence on phase-noise severity, indicating that CPTR's coefficient-domain efficiency is largely insensitive to the Wiener innovation variance. Fig.~\ref{fig:results}(d) further demonstrates the predicted $\mathcal{O}(\sigma_\Delta^2N)$ scaling of the residual approximation floor across bandwidths and linewidth coefficients.

Fig.~\ref{fig:results}(e) shows trajectory MSE versus polynomial order. MSE decreases sharply from $P=0$ to $P=3$ as the polynomial surrogate captures progressively richer intra-symbol phase dynamics. For $P>3$, the MSE increases due to amplified estimation noise and ill-conditioning of the Vandermonde matrix under limited pilot density. This behavior reflects the bias--variance trade-off of polynomial phase reconstruction, with $P=3$ providing the best balance between approximation accuracy and numerical robustness. Increasing $P$ improves approximation accuracy but worsens Vandermonde conditioning, increasing LS sensitivity to pilot noise.

Fig.~\ref{fig:results}(f) shows uncoded BER for 16-QAM. Without compensation, NC and CPE exhibit substantial BER floors caused by intra-symbol phase drift, while linear interpolation, spline interpolation, and EKF approach the AWGN bound at high SNR. CPTR achieves BER~$=4.7\times10^{-3}$ at SNR~$=16$~dB, a $5.0\times$ reduction relative to CPE ($2.36\times10^{-2}$) at identical pilot overhead, but exhibits a residual floor of approximately $2.3\times10^{-4}$ at SNR~$\ge20$~dB, consistent with~\eqref{eq:floor}. This behavior reflects the trade-off between closed-form low-complexity estimation and surrogate-model mismatch.
 
\section{Conclusion}
\label{sec:conclusion}

This letter presented CPTR, a closed-form intra-symbol phase-noise tracking method for THz OFDM systems. The key insight is that mitigating WPN at THz carriers requires continuous sample-level phase reconstruction within each OFDM symbol, rather than conventional per-symbol CPE correction. By modeling the phase trajectory with a low-order polynomial and solving a single least-squares system of dimension $(P{+}1)\!\times\!(P{+}1)$, CPTR achieves $\mathcal{O}(N_p+N)$ complexity while remaining within 0.2~dB of the CRB in the polynomial coefficient space. At 300~GHz with $N=1024$ subcarriers and 16-QAM, CPTR attains substantially lower complexity than EKF-based methods and significant BER improvement over conventional CPE compensation. A limitation of the proposed framework is the residual high-SNR error floor caused by the finite-order polynomial approximation of the Wiener phase trajectory. Extending the framework to mixed Wiener--Gaussian phase-noise models and exploring adaptive basis design constitute directions for future work.

\appendix
\textit{Proof of Proposition~\ref{prop:approx}:} The WPN trajectory
$\phi[n]=\phi[0]+\sum_{i=1}^{n}\delta[i]$
is a zero-drift Wiener process on $[0,N-1]$ with covariance
$K(n,m)=\min(n,m)\sigma^{2}_{\Delta}$.
Its Karhunen--Lo\`eve (KL) expansion over a finite interval
admits eigenvalues satisfying
$\lambda_p=\mathcal{O}(p^{-2})$~\cite{Papoulis2002}.
Consequently, the truncation error of the optimal rank-$(P+1)$
KL approximation satisfies
\begin{equation}\label{eq:kl_err}
  \sigma^{2}_{\varepsilon,\mathrm{KL}}
  =\sum_{p=P+1}^{\infty}\lambda_p
  \;\lesssim\;
  C(P+1)^{-1}\sigma^{2}_{\Delta}N,
\end{equation}
which decays algebraically as $\mathcal{O}(P^{-1})$.

Since the Legendre and monomial bases span the same degree-P subspace, their projections are related by a nonsingular linear transformation; norm equivalence in finite dimensions implies the projection errors differ only by an N-independent conditioning constant. The monomial projection thus preserves the $\mathcal{O}(P^{-1})$ decay of the KL truncation error.



\begin{thebibliography}{34}

\bibitem{Han2022}
C. Han et al., ``Terahertz Wireless Channels: A Holistic Survey on Measurement, Modeling, and Analysis,'' \textit{IEEE Commun. Surv. Tutor.}, vol. 24, no. 3, pp. 1670--1707, thirdquarter 2022.

\bibitem{Sharma2025}
S. Sharma, P. K. Singya, K. Deka, C. Adjih, and M. Sharma, ``Terahertz Communication: State-of-the-Art and Future Directions,'' \textit{IEEE Open J. Commun. Soc.}, vol. 6, pp. 6281--6322, 2025.

\bibitem{Akyildiz2020}
I. F. Akyildiz, A. Kak, and S. Nie, ``6G and Beyond: The Future of Wireless Communications Systems,'' \textit{IEEE Access}, vol. 8, pp. 133995--134030, 2020.

\bibitem{Rappaport2019}
T. S. Rappaport et al., ``Wireless Communications and Applications Above 100 GHz: Opportunities and Challenges for 6G and Beyond,'' \textit{IEEE Access}, vol. 7, pp. 78729--78757, 2019.

\bibitem{Sarieddeen2021}
H. Sarieddeen, M. -S. Alouini, and T. Y. Al-Naffouri, ``An Overview of Signal Processing Techniques for Terahertz Communications,'' \textit{Proc. IEEE}, vol. 109, no. 10, pp. 1628--1665, Oct. 2021.

\bibitem{Khanzadi2014}
M. R. Khanzadi, D. Kuylenstierna, A. Panahi, T. Eriksson, and H. Zirath, ``Calculation of the Performance of Communication Systems From Measured Oscillator Phase Noise,'' \textit{IEEE Trans. Circuits Syst. I, Reg. Papers.}, vol. 61, no. 5, pp. 1553--1565, May 2014.

\bibitem{Bicais2019}
S. Bicais and J. -B. Dore, ``Phase Noise Model Selection for Sub-THz Communications,'' in \textit{Proc. 2019 IEEE Global Communications Conference (GLOBECOM)}, Waikoloa, HI, USA, 2019, pp. 1--6.

\bibitem{Pollet1995}
T. Pollet, M. Van Bladel, and M. Moeneclaey, ``BER sensitivity of OFDM systems to carrier frequency offset and Wiener phase noise,'' \textit{IEEE Trans. Commun.}, vol. 43, no. 2/3/4, pp. 191--193, Feb./March/April 1995.

\bibitem{Wu2004}
S. Wu and Y. Bar-Ness, ``OFDM systems in the presence of phase noise: consequences and solutions,'' \textit{IEEE Trans. Commun.}, vol. 52, no. 11, pp. 1988--1996, Nov. 2004.

\bibitem{Petrovic2007}
D. Petrovic, W. Rave, and G. Fettweis, ``Effects of Phase Noise on OFDM Systems With and Without PLL: Characterization and Compensation,'' \textit{IEEE Trans. Commun.}, vol. 55, no. 8, pp. 1607--1616, Aug. 2007.

\bibitem{Stefanatos2017}
S. Stefanatos, F. Foukalas, and T. Khattab, ``On the Achievable Rates of OFDM With Common Phase Error Compensation in Phase Noise Channels,'' \textit{IEEE Trans. Commun.}, vol. 65, no. 8, pp. 3509--3521, Aug. 2017.

\bibitem{Corvaja1997}
R. Corvaja and S. Pupolin, ``Phase noise effects in QAM systems,'' in \textit{Proc. 8th International Symposium on Personal, Indoor and Mobile Radio Communications - PIMRC '97}, Helsinki, Finland, 1997, pp. 452-456 vol.2.

\bibitem{Giannakis1998}
G. B. Giannakis and C. Tepedelenlioglu, ``Basis expansion models and diversity techniques for blind identification and equalization of time-varying channels,'' \textit{Proc. IEEE}, vol. 86, no. 10, pp. 1969--1986, Oct. 1998.

\bibitem{Do2025}
H. Do, N. Lee and A. Lozano, ``Multidimensional Polynomial Phase Estimation,'' \textit{IEEE open j. signal process.}, vol. 6, pp. 651--681, 2025.

\bibitem{Rabiei2010}
P. Rabiei, W. Namgoong, and N. Al-Dhahir, ``A Non-Iterative Technique for Phase Noise ICI Mitigation in Packet-Based OFDM Systems,'' \textit{IEEE Trans. Signal Process.}, vol. 58, no. 11, pp. 5945--5950, Nov. 2010.

\bibitem{Mehrpouyan2012}
H. Mehrpouyan, A. A. Nasir, S. D. Blostein, T. Eriksson, G. K. Karagiannidis, and T. Svensson, ``Joint Estimation of Channel and Oscillator Phase Noise in MIMO Systems,'' \textit{IEEE Trans. Signal Process.}, vol. 60, no. 9, pp. 4790-4807, Sep. 2012.

\bibitem{Mathecken2017}
P. Mathecken, T. Riihonen, S. Werner, and R. Wichman, ``Constrained Phase Noise Estimation in OFDM Using Scattered Pilots Without Decision Feedback,'' \textit{IEEE Trans. Signal Process.}, vol. 65, no. 9, pp. 2348--2362, May 2017.

\bibitem{Chen2025}
D. Chen, L. Song, P. Liu, K. Luo, W. Peng, and W. Wang, ``Phase Noise Estimation and Pilot Design Suppressing Intrinsic Interference for mmWave FBMC-OQAM Systems,'' \textit{IEEE Trans. Commun.}, vol. 73, no. 11, pp. 12087--12099, Nov. 2025.

\bibitem{Chang2025}
Z. Chang, Y. Xu, J. Chen, N. Xie, Y. He, and H. Li, ``Modeling, Estimation, and Applications of Phase Noise in Wireless Communications: A Survey,'' \textit{IEEE Commun. Surv. Tutor.}, vol. 27, no. 2, pp. 912--940, Apr. 2025.

\bibitem{Bello2023}
Y. Bello, J. -B. Doré, and D. Demmer, ``Wiener Interpolation Filter for Phase Noise Estimation in sub-THz Transmission,'' in \textit{Proc. IEEE 97th Veh. Technol. Conf. (VTC2023-Spring)}, Florence, Italy, 2023, pp. 1--5.

\bibitem{Tarboush2021}
S. Tarboush et al., ``TeraMIMO: A Channel Simulator for Wideband Ultra-Massive MIMO Terahertz Communications,'' \textit{IEEE Trans. Veh. Technol.}, vol. 70, no. 12, pp. 12325--12341, Dec. 2021.

\bibitem{Kay1993}
S.~M.~Kay, {\it{Fundamentals of Statistical Signal Processing, Vol.~I: Estimation Theory}}. NJ, USA, Prentice Hall, 1993.

\bibitem{Trees2001}
H. L. Van Trees, {\it{Detection, Estimation, and Modulation Theory, Part~I}}. NJ, USA, Wiley-Interscience, 2001.

\bibitem{Papoulis2002}
A.~Papoulis and S.~U.~Pillai, {\it{Probability, Random Variables, and Stochastic Processes}}, 4th~ed. New York, USA, McGraw-Hill, 2002.

\end{thebibliography}
\end{document}